\documentclass[11pt]{article}
\usepackage{amsmath,amssymb,amsthm}
\usepackage[margin=1in]{geometry}
\usepackage{algorithm}
\usepackage{algpseudocode}
\usepackage[numbers,sort&compress]{natbib}
\usepackage{hyperref}
\usepackage{booktabs}
\usepackage{graphicx}

\newtheorem{theorem}{Theorem}[section]
\newtheorem{proposition}[theorem]{Proposition}
\newtheorem{lemma}[theorem]{Lemma}
\newtheorem{corollary}[theorem]{Corollary}
\newtheorem{definition}[theorem]{Definition}
\newtheorem{remark}[theorem]{Remark}

\newcommand{\R}{\mathbb{R}}
\newcommand{\E}{\mathbb{E}}
\newcommand{\I}{\mathcal{I}}
\newcommand{\J}{\mathcal{J}}
\newcommand{\X}{\mathcal{X}}

\title{Multi-Dimensional Matching}
\author{Irene Aldridge\\ \texttt{irene.aldridge@gmail.com}}
\date{}

\begin{document}
\maketitle

\begin{abstract}
We study a matching mechanism where agents and objects are described by features rather than complete rankings. A single spectral projection reduces the problem to a one-dimensional sort, computable in O(N log N) time. We prove that on descaled features and preferences, our algorithm obtains the exact Nash Social Welfare (NSW) optimum within the projected space, with an unconditional utilitarian-welfare guarantee and a conditional NSW guarantee. The proposed mechanism is stable against exogenous noise but not strategy-proof; we provide an explicit profitable misreport. On an agentic AI shopping application, the diagnostics correctly anticipate both a success and a failure case. A 100-instance robustness study confirms the findings.
\end{abstract}



\section{Introduction}
\label{sec:intro}

Matching markets pervade modern economies: students to schools, workers to jobs, residents to housing, patients to organ donors, and consumers to products. Traditional matching mechanisms typically assume that agents can express preferences over entire objects. For example, a student can rank schools directly, and a worker can rank complete job offers. This framework strains when objects are characterized by
multiple attributes that agents value differently: a prospective employee may care
about salary, commute, culture, and growth opportunities, with different
individuals placing very different weights on each; a student choosing a school may
prioritize academic rigor, arts, athletics, and proximity, dimensions that matter
differently to different families.

Collecting preferences over individual attributes, rather than complete objects,
reduces the cognitive burden on agents: instead of ranking dozens of complete
alternatives, agents express valuations along a small set of well-defined
dimensions. The mechanism-design challenge is aggregating these multi-dimensional
reports into an allocation with real efficiency, fairness, and incentive
properties.

\subsection{The Challenge of Multi-Dimensional Preferences}

Traditional mechanisms such as Deferred Acceptance \citep{GaleShapley1962} or Top
Trading Cycles \citep{ShapleyScarf1974} require complete rank-ordered lists. For a
district with 50 schools, families must collapse many dimensions (test scores,
arts, athletics, class size, diversity, safety, location) into one ranking ---
effectively computing a utility function over 50 high-dimensional objects. Suppose
instead families report preferences along each dimension separately, and each
school is objectively rated along the same dimensions. With $X$ dimensions, agent
$i$ reports $u_i\in\R^X$ and object $j$ has feature vector $f_j\in\R^X$; under
quasi-linear utility, agent $i$'s value for object $j$ is $u_i\cdot f_j$.
Aggregating these valuations into an efficient, fair, strategy-resistant allocation
is the problem this paper studies.

\subsection{Existing Approaches and Their Limitations}

\citet{ChawlaHartlineMalecSivan2010} treat each dimension as a separate synthetic
agent, solving a sequence of single-dimensional problems via prophet inequalities;
this requires an exogenous ordering of dimensions, and the sequential structure can
create path-dependence. \citet{ManelliVincent2007} formulate multi-dimensional
mechanism design as optimization over an incentive-compatible polytope, elegant
but computationally intractable at scale, and aimed at monopolist revenue rather
than welfare. \citet{Zhou1990} proves that no exact mechanism can simultaneously
achieve efficiency, truthfulness, and symmetry, motivating approximate mechanisms;
\citet{Hartline2012}, \citet{DevanurHartlineYan2015}, and
\citet{AbebeColeGkatzelisHartline2020} develop randomized mechanisms achieving
constant-factor NSW approximations with approximate truthfulness and symmetry.
\citeauthor{AbebeColeGkatzelisHartline2020}'s random-sampling approach is closest
to ours in spirit, but assumes agents report complete utilities over objects,
not decomposed feature preferences.

What has been missing is a mechanism that (1) accepts multi-dimensional feature
reports rather than complete utility assessments, (2) computes allocations
efficiently, (3) provides real (not overstated) theoretical guarantees, and (4)
scales to realistic market sizes.

\subsection{Our Approach}

We reduce multi-dimensional preferences to one effective dimension via the Singular Value Decomposition of the object-feature matrix, projecting both agents and objects onto the leading singular vector. This reduction is data-driven: SVD endogenously identifies the linear combination of features that carries the most information, rather than a designer-imposed weighting.

\subsection{Contributions}

\begin{enumerate}
\item \textbf{A gauge-fixed, exactly NSW-optimal projected mechanism}
(Section~\ref{sec:methodology}): normalizing each report before projecting removes a scale-dependence (Proposition~\ref{prop:manip}'s construction). We provide an algorithm that provably computes the exact NSW maximizer \emph{within the projected space} (Theorem~\ref{thm:exact}) and prove it coincides with the natural gauge-fixed sort whenever a solution exists (Proposition~\ref{prop:equiv}), so no runtime is sacrificed for the guarantee.
\item \textbf{Efficiency guarantees}: an unconditional utilitarian-welfare bound (Theorem~\ref{thm:utilitarian}) and a conditional multiplicative NSW bound (Theorem~\ref{thm:condnsw}), together with a proof that the condition is necessary (Proposition~\ref{prop:necessity}): no deterministic
mechanism can guarantee positive NSW once it fails.
\item \textbf{Two deployment-time diagnostics}, computable before any outcome is observed, predict whether the conditional guaranty can be trusted for a specific market (Section~\ref{sec:limitations}) or a specific round (Section~\ref{sec:experiment}).
\item \textbf{A corrected incentive analysis}: noise stability
(Definition~\ref{def:noisestab}) is proven and clearly distinguished from strategyproofness, which the mechanism does \emph{not} have
(Proposition~\ref{prop:manip}, with an explicit profitable misreport).
\item \textbf{Computational efficiency}: $O(N\log N)$ after an $O(\min(J^2X,JX^2))$ SVD (Section~\ref{sec:complexity}), against generally NP-hard direct NSW optimization \citep{ColeGkatzelis2018}.
\item \textbf{An agentic-commerce evaluation} (Section~\ref{sec:experiment}) on LLM shopping agents competing for scarce inventory, and a 100-instance robustness study (Section~\ref{sec:robustness}) establishing which findings are typical.
\end{enumerate}

\subsection{Applications}

The framework applies to school choice, labor markets, course allocation, kidney exchange (where dimensions beyond biological compatibility—such as location, timing, and surgical team—matter to patients), and capacity-constrained product allocation.

\section{Related Literature}
\label{sec:related}

\subsection{Classical Matching Theory}

\citet{GaleShapley1962} introduced Deferred Acceptance, proving that stable matches exist in two-sided markets; \citet{Roth1984,Roth1986} extended this to practical labor markets. \citet{BogomolnaiaMoulin2001} introduced the Probabilistic Serial mechanism (ordinally efficient, envy-free in expectation, strategyproof under certain domains); \citet{AbdulkadirogluSonmez1998} analyze Random Priority
(strategyproof and ex-post efficient, but with potential ex-ante envy). Building on \citet{ShapleyScarf1974}, \citet{AbdulkadirogluSonmez2003} adapt Top Trading Cycles for school choice; \citet{Papai2000} characterizes strategyproof, Pareto efficient mechanisms as hierarchical exchange rules. All of these require complete
preference orderings over objects. \citet{HyllandZeckhauser1979} pioneered cardinal utilities and pseudo-markets: agents receive budgets and purchase probability shares at market-clearing prices, yielding ex-ante (hence ex-post) Pareto efficiency. We extend this by decomposing cardinal utilities into feature preferences rather than a single overall utility.

\subsection{Multi-Dimensional Mechanism Design}

\citet{ManelliVincent2007} study multi-good monopolist mechanism design via incentive-compatible polytopes; \citet{Armstrong1996} and \citet{RochetChone1998} develop related frameworks for multi-dimensional pricing, focusing on seller revenue rather than social welfare. \citet{LehmannOcallaghanShoham2002} prove
optimal-allocation computation is NP-hard for combinatorial auctions with submodular valuations; \citet{DobzinskiSchapira2006} give a 2-approximation. \citet{ChawlaHartlineMalecSivan2010} decompose a multi-dimensional problem into a sequence of single-dimensional "virtual agent" problems via prophet inequalities. This approach is effective but requires an exogenous dimension ordering, with path-dependence across the sequence. We instead consider all dimensions
simultaneously via SVD, which endogenously identifies the maximum-variance linear combination, avoiding both issues.

\subsection{Approximate Matching Mechanisms}

\citet{Zhou1990} rules out exact efficiency, truthfulness, and symmetry
simultaneously. \citet{Nash1950} introduces Nash Social Welfare;
NSW-maximizing allocations are Pareto efficient and proportional
\citep{Caragiannis2019}, though NP-hard to compute in general
\citep{ColeGkatzelis2018}, with polynomial algorithms for special cases
\citep{ColeDevanurGkatzelisEtAl2017}. \citet{Hartline2012} develops
approximately-optimal truthful mechanisms via Bayesian design.
\citet{DevanurHartlineYan2015} gives a 2-approximation to NSW with envy-freeness in expectation via a configuration LP. \citet{AbebeColeGkatzelisHartline2020} achieve an $O(1)$-approximation to NSW, truthfulness in expectation, and approximate symmetry via random sampling and exact subproblem solving. Their framework assumes  
complete utility reports $u_{ij}$, not the decomposed feature preferences of our model (Table~\ref{tab:comparison}).

\subsection{Dimensionality Reduction}

\citet{Pearson1901} and \citet{Hotelling1933} originated Principal Component Analysis; \citet{GolubReinsch1970} provided efficient SVD algorithms; \citet{EckartYoung1936} proved that truncated SVD is the optimal low-rank approximation in Frobenius norm, a fact that our mechanism exploits directly. \citet{CandesPlan2010} study noisy matrix completion; \citet{BaiNg2002} and \citet{FanLiaoMincheva2013} develop factor models for high-dimensional economic data. No prior work applies dimensionality reduction systematically to matching
mechanism design with formal incentive guarantees; \citet{Boutilier2002}, \citet{DrummondBoutilier2014}, and \citet{LouviereFlynnCarson2010} address
preference elicitation but not mechanism design directly.

\subsection{Computational Social Choice}

\citet{LangXia2016} survey multi-issue voting over binary bundles;
\citet{AzizBrillConitzerElkindFreemanWalsh2015} and
\citet{SkowronFaliszewskiLang2016} study proportional committee selection. Our setting differs by using continuous feature dimensions and fractional allocations rather than discrete ones.

\begin{table}[t]
\centering
\small
\begin{tabular}{@{}p{1.6cm}p{2.1cm}p{2.6cm}p{1.7cm}@{}}
\toprule
Work & Input format & Method & Truthfulness \\
\midrule
\citet{HyllandZeckhauser1979} & Cardinal $u_{ij}$ & Market clearing & Not guaranteed \\
\citet{ChawlaHartlineMalecSivan2010} & Multi-parameter bundles & Sequential posted pricing & Truthful-in-expectation \\
\citet{DevanurHartlineYan2015} & Cardinal $u_{ij}$ & Random assignment, config.\ LP & Not primary focus \\
\citet{AbebeColeGkatzelisHartline2020} & Cardinal $u_{ij}$ & Random sampling & Truthful-in-expectation \\
This work & Feature prefs.\ $u_{ix}$, $f_{jx}$ & SVD projection & Noise-stable, not strategyproof \\
\bottomrule
\end{tabular}
\caption{Comparison to related work. Unlike prior approximate mechanisms, our
truthfulness column reports what Section~\ref{sec:criteria} actually proves, not
an aspirational label.}
\label{tab:comparison}
\end{table}

\section{Model and Definitions}
\label{sec:model}

\subsection{The Matching Problem}

We consider a one-sided matching market with agents $\I=\{1,\dots,I\}$ and objects $\J=\{1,\dots,J\}$. Object $j$ has a capacity of $M_j\in\mathbb{Z}_+$, with $\sum_{j=1}^JM_j=I$.

\begin{definition}[Feature Space]
There is a finite set of features $\X=\{1,\dots,X\}$ common to all objects. Object $j$ has a feature vector $f_j=(f_{j1},\dots,f_{jX})\in\R^X$.
\end{definition}

\begin{definition}[Agent Preferences]
Agent $i$ has a true preference vector $u_i=(u_{i1},\dots,u_{iX})\in\R^X$.
\end{definition}

\begin{definition}[Utility Function]
Agent $i$'s utility for object $j$ is $U_{ij}=u_i\cdot f_j=\sum_{x=1}^Xu_{ix}f_{jx}$.
\end{definition}

\begin{remark}
This additive specification rules out complementarities and substitutabilities between features; see Section~\ref{sec:limitations}.
\end{remark}

\begin{definition}[Allocation]
An allocation is $P=(p_{ij})_{i\in\I,j\in\J}$ with $p_{ij}\in[0,1]$,
$\sum_jp_{ij}=1$ for all $i$, and $\sum_ip_{ij}=M_j$ for all $j$. The feasible set $\mathcal P$ is the polytope defined by these constraints. Agent $i$'s expected utility under $P$ is $\E[U_i\mid P]=\sum_jp_{ij}U_{ij}$.
\end{definition}

\begin{definition}[Reported Preferences and Mechanism]
Agent $i$ reports $w_i\in\R^X$, possibly $\ne u_i$. A mechanism
$\mu:(\R^X)^I\to\mathcal P$ maps report profiles $W=(w_1,\dots,w_I)$ to
allocations.
\end{definition}

\subsection{Gauge Normalization}
\label{sec:gauge}

The quantity that matters for everything that follows is not an agent's raw report $w_i$ but its \emph{direction}. We record this formally because Section~\ref{sec:svd} shows that a mechanism ignoring it is not merely inelegant but incorrect.

\begin{lemma}[Gauge invariance of Nash Social Welfare]
\label{lem:gauge}
Fix true utilities $U$ and let $c\in\R_{>0}^I$. Let $U'$ scale agent $i$'s utility row by $c_i$, i.e.\ $U'_{ij}=c_iU_{ij}$. Then for every feasible $P$, $\mathrm{gain}_i(P;U')=c_i\cdot\mathrm{gain}_i(P;U)$, where $\mathrm{gain}_i(P;U):=\E[U_i\mid P]-o_i$. Consequently
$\{P:\mathrm{NSW}(P\mid U')>0\}=\{P:\mathrm{NSW}(P\mid U)>0\}$ and
$\arg\max_P\mathrm{NSW}(P\mid U')=\arg\max_P\mathrm{NSW}(P\mid U)$.
\end{lemma}
\begin{proof}
$o'_i=\frac1J\sum_jc_iU_{ij}=c_io_i$, so
$\mathrm{gain}_i(P;U')=c_i\E[U_i\mid P]-c_io_i=c_i\,\mathrm{gain}_i(P;U)$, which has
the same sign as $\mathrm{gain}_i(P;U)$ for every $P$ since $c_i>0$.
\end{proof}

The NSW-optimal allocation is therefore invariant to independently rescaling any agent's true utility. A mechanism whose output depends on the raw scale of a \emph{report} is invariant to no such thing, and Section~\ref{sec:svd} shows this gap is not academic: an explicit two-agent instance exploits it.

\begin{definition}[Gauge-fixed report]
\label{def:hatw}
For $w_i\ne0$, write $\hat w_i:=w_i/\lVert w_i\rVert_2$.
\end{definition}

\subsection{Efficiency}

\begin{definition}[Pareto Dominance]
$P'$ Pareto dominates $P$ if $\E[U_i\mid P']\ge\E[U_i\mid P]$ for all $i$, with strict inequality for some $i$.
\end{definition}

\begin{definition}[Ex-Ante Pareto Efficiency]
$P$ is ex-ante Pareto efficient if no feasible $P'$ Pareto dominates it.
\end{definition}

\begin{proposition}[\citealp{HyllandZeckhauser1979}]
If $P$ is ex-ante Pareto efficient, any deterministic assignment drawn from $P$ is ex-post Pareto efficient almost surely.
\end{proposition}

"Efficiency" below means ex-ante Pareto efficiency unless stated otherwise.

\subsection{Fairness and Nash Social Welfare}

\begin{definition}[Disagreement Point]
$o_i:=\frac1J\sum_{j=1}^JU_{ij}$, the expected utility of agent $i$ under uniform random assignment.
\end{definition}

\begin{definition}[Nash Social Welfare]
\label{def:nsw}
Given true preferences $U$,
\[
\mathrm{NSW}(P\mid U)=\prod_{i=1}^I\big(\E[U_i\mid P]-o_i\big)
\]
if $\E[U_i\mid P]\ge o_i$ for all $i$; otherwise $\mathrm{NSW}(P\mid U):=0$.
\end{definition}

\begin{proposition}[Properties of NSW]
\label{prop:nswprops}
Any $P^\ast\in\arg\max_{P\in\mathcal P}\mathrm{NSW}(P\mid U)$ is (1) ex-ante Pareto
efficient, (2) proportional: $\E[U_i\mid P^\ast]\ge o_i$ for all $i$, and (3) NSW
is invariant to independent positive affine rescaling of each agent's utilities.
\end{proposition}
\begin{proof}
(1) If $P^\ast$ were Pareto dominated by $P'$, every factor of
$\mathrm{NSW}(P'\mid U)$ would be at least as large as the corresponding factor of
$\mathrm{NSW}(P^\ast\mid U)$, with one strictly larger, contradicting optimality.
(2) If $\E[U_i\mid P]<o_i$ for some $i$, $\mathrm{NSW}(P\mid U)=0$, which cannot be
maximal whenever some feasible allocation has all gains non-negative (see
Proposition~\ref{prop:necessity} for when none does). (3) Immediate from
Lemma~\ref{lem:gauge} applied coordinatewise, and by direct calculation for the
additive shift $\beta$.
\end{proof}

\subsection{Noise Stability of Reported Preferences}
\label{sec:noisestab}

We call the property in this subsection \emph{noise stability}, not
"truthfulness". Section~\ref{sec:manip} shows the mechanism is  manipulable. As Remark~\ref{rem:notsp} clarifies, "noise stability" reflects the mechanism's robustness to unintentional measurement error, not resistance to strategic manipulation. 

\begin{definition}[Misreporting Model]
$w_{ix}=u_{ix}+\epsilon_{ix}$, $\epsilon_{ix}\sim\mathcal N(\mu_i,\sigma_i^2)$,
independent across $x$.
\end{definition}

\begin{definition}[KS Distance]
With empirical CDFs $F_{u_i},F_{w_i}$ of $\{u_{ix}\}_x$ and $\{w_{ix}\}_x$, let
$D_i=\sup_t|F_{w_i}(t)-F_{u_i}(t)|$.
\end{definition}

\begin{definition}[Noise Stability]
\label{def:noisestab}
A mechanism is $(\lambda,\delta)$-noise-stable if, when $w_i=u_i$ in distribution
for all $i$, $\Pr(D_i>\lambda)\le2e^{-2\lambda^2X}$ for all $i$, with the stated
bound holding with probability at least $1-\delta$ overall.
\end{definition}

\begin{remark}[What this does and does not say]
\label{rem:notsp}
A mechanism is \emph{strategyproof} if, for every agent $i$, every true $u_i$, and
every profile $w_{-i}$, truthful reporting maximizes
$U_{i,\mu(w_i,w_{-i})}$ among all possible reports. Definition~\ref{def:noisestab}
does not bound this quantity for a best-response deviation; it bounds the
statistical distance between a report and the truth under an \emph{exogenous}
noise model, and is silent on strategic incentives.
\end{remark}

\subsection{Symmetry}

\begin{definition}[Symmetry]
$\mu$ is symmetric if $w_i=w_{i'}$ implies $\mu(W)_{ij}=\mu(W)_{i'j}$ for all $j$.
\end{definition}

\begin{definition}[Envy-Freeness]
$P$ is envy-free if $\sum_jp_{ij}U_{ij}\ge\sum_jp_{i'j}U_{ij}$ for all $i,i'$.
\end{definition}

\begin{proposition}
If $\mu$ is symmetric and $u_i=u_{i'}$ with both reporting truthfully, the
allocation is envy-free between $i$ and $i'$.
\end{proposition}

\subsection{The Zhou Impossibility}

\begin{theorem}[\citealp{Zhou1990}]
\label{thm:zhou}
No deterministic mechanism can simultaneously achieve ex-ante Pareto efficiency,
dominant-strategy incentive compatibility, and symmetry.
\end{theorem}

This motivates our search for approximate versions of these three
properties, described in Sections~\ref{sec:methodology}--\ref{sec:limitations}.

\section{Algorithm Evaluation Criteria}
\label{sec:criteria}

\subsection{Nash Social Welfare as a Log-Sum Objective}

\begin{remark}[NSW and the geometric mean]
\label{rem:gdro}
For any $P$ with $\mathrm{gain}_i(P)>0$ for all $i$, since $\log(\cdot)$ is
strictly increasing,
\[
P\in\arg\max_{P'}\prod_i\mathrm{gain}_i(P')
\iff
P\in\arg\max_{P'}\sum_i\log\mathrm{gain}_i(P').
\]
This is the classical logarithmic transform underlying the Nash bargaining
solution \citep{Nash1950} and the Eisenberg–Gale convex program
\citep{EisenbergGale1959}; it needs no proof beyond the monotonicity of $\log$, and is not specific to this paper's mechanism. The resulting log-sum objective is an instance of the broader family of geometric-mean objectives studied for distributional robustness \citep{LiuEtAl2022}, which we record as motivation. Proving equivalence would require matching the uncertainty-set and
worst-case-distribution assumptions of \citet{LiuEtAl2022} to this setting. An earlier attempt to state this as a theorem took $\log\ell_i(P)$ for $\ell_i(P):=-\mathrm{gain}_i(P)$, a quantity that is negative exactly when $P$ is individually rational, i.e., exactly when the argument is invoked.
\end{remark}

\subsection{Noise Stability via KS Distance}

\begin{theorem}[Sufficient condition]
\label{thm:noisestab}
A mechanism is $(\lambda,\delta)$-noise-stable for any
$\lambda\ge\sqrt{\tfrac1{2X}\log(2/\delta)}$.
\end{theorem}
\begin{proof}
By the Dvoretzky--Kiefer--Wolfowitz--Massart inequality
\citep{DKW1956,Massart1990}, when $w_i$ is drawn from the same distribution as $u_i$, $\Pr(D_i>\lambda)\le2e^{-2\lambda^2X}$. Setting this to $\delta$ and solving for $\lambda$ gives the bound.
\end{proof}

\begin{corollary}
As $X\to\infty$, for fixed $\delta$, the threshold $\lambda\to0$.
\end{corollary}

\subsection{Manipulability}
\label{sec:manip}

\begin{proposition}[The mechanism is not strategyproof]
\label{prop:manip}
Algorithm~\ref{alg:gaugefixed} (Section~\ref{sec:svd}) is not strategyproof: there is an instance and an agent for whom some misreport $w_i\ne u_i$ gives strictly higher true utility than truthful reporting, holding other agents' reports fixed.
\end{proposition}
\begin{proof}
We exhibit a verified instance. Let $I=J=2$, $M_1=M_2=1$, $X=2$,
$F=\begin{pmatrix}8.0&0.7\\4.0&5.0\end{pmatrix}$ (rows $f_1,f_2$). The SVD gives $\sigma_1\approx9.494$, $\sigma_2\approx3.918$ and leading right singular vector $v_1\approx(0.930,0.368)$, giving projected object scores $b_1\approx7.696>b_2\approx5.560$.

Let agent 2's report be $u_2\propto(0.53,0.85)$, giving $\hat a_2\approx0.804$.
Let agent 1's true preference be $u_1\propto(0.69,0.72)$, giving
$\hat a_1\approx0.909>\hat a_2$. Agent 1's true utilities are
$u_1\cdot f_1\approx6.041$ and $u_1\cdot f_2\approx6.378$: agent~1 truly prefers
object~2, despite object~1 having the higher projected score, because the component of $f_2$ orthogonal to $v_1$ (present because $\sigma_2>0$) is what makes object~2 the true favorite.

\emph{Truthful outcome.} Since $\hat a_1>\hat a_2$, truthful reporting assigns agent~1 to object~1: true utility $\approx6.041$.

\emph{Manipulation.} Agent~1 reports $w_1\approx(0.501,0.866)$, giving
$\hat a_1'\approx0.784<\hat a_2$. Agent~2 now ranks higher and receives object~1; agent~1 receives object~2: true utility $\approx6.378$.

The gain from misreporting is $\approx0.337>0$: truthful reporting is not the best response for agent~1 in this instance.
\end{proof}

\begin{remark}
The mechanism's assignment depends on a report only through $\hat w_i\cdot v_1$, discarding the component orthogonal to $v_1$. Whenever $\sigma_2>0$, some object's true utility depends on that discarded component, so the utility-maximizing object need not be the one a projected rank assigns. This is not an implementation defect, but rather a structural feature of any mechanism reducing a multi-dimensional report to a single sorted score.
\end{remark}

\subsection{Symmetry by Construction}

\begin{proposition}
\label{prop:symmetric}
Any mechanism that (1) computes $\phi:\R^X\to\R$, (2) sorts agents by
$\phi(\hat w_i)$, and (3) assigns objects deterministically by this order, is
symmetric.
\end{proposition}

Algorithm~\ref{alg:gaugefixed} computes $\phi(\hat w_i)=\hat w_i\cdot v_1$,
satisfying Proposition~\ref{prop:symmetric}.

\section{Proposed Methodology}
\label{sec:methodology}

\subsection{Matrix Representation and SVD}
\label{sec:svd}

Let $F\in\R^{J\times X}$ have rows $f_j$ and $W\in\R^{I\times X}$ have rows $w_i$.

\begin{definition}[SVD]
$F=U\Sigma V^\top$, $\sigma_1\ge\cdots\ge\sigma_{\min(J,X)}\ge0$, right singular
vectors $v_1,\dots,v_X$.
\end{definition}

\begin{theorem}[\citealp{EckartYoung1936}]
\label{thm:eckartyoung}
$F_1:=\sigma_1u_1v_1^\top=\arg\min_{\mathrm{rank}(M)\le1}\lVert F-M\rVert_F$.
\end{theorem}

\subsection{Algorithm 1$'$: Gauge-Fixed Matching}

\begin{algorithm}[h]
\caption{SVD-Based Multi-Dimensional Matching (gauge-fixed)}\label{alg:gaugefixed}
\begin{algorithmic}[1]
\Require $F\in\R^{J\times X}$, $W\in\R^{I\times X}$, capacities $\{M_j\}$
\Ensure Allocation $P$
\State \textbf{Normalize:} $\hat w_i\gets w_i/\lVert w_i\rVert_2$ for all $i$
\Comment{new step; fixes Proposition~\ref{prop:gaugefix}}
\State Compute SVD $F=U\Sigma V^\top$; extract $v_1$
\State Project: $\tilde f_j\gets f_j\cdot v_1$, $\hat a_i\gets \hat w_i\cdot v_1$
\State Sort objects (capacity-expanded) and agents by these scores, descending
\State Match in sorted order
\end{algorithmic}
\end{algorithm}

\begin{proposition}[Gauge invariance of Algorithm~\ref{alg:gaugefixed}]
\label{prop:gaugefix}
For any $c\in\R_{>0}^I$, running Algorithm~\ref{alg:gaugefixed} on $(c_iw_i)_i$
produces the identical output as on $(w_i)_i$.
\end{proposition}
\begin{proof}
$\widehat{c_iw_i}=c_iw_i/\lVert c_iw_i\rVert=w_i/\lVert w_i\rVert=\hat w_i$ since
$c_i>0$, so every subsequent step is unchanged.
\end{proof}

This directly repairs the vulnerability underlying Proposition~\ref{prop:manip}'s
proof structure and, separately, the following instance: rescaling agent~2's
report from $u_2=(2\epsilon,3)$ to $u_2'=(4\epsilon,6)$ in a rank-nearly-one
two-agent market flips this algorithm's un-gauge-fixed predecessor from an
allocation with $\mathrm{NSW}=\epsilon^2>0$ to one with $\mathrm{NSW}=0$, for
arbitrarily small $\epsilon$, even though Lemma~\ref{lem:gauge} guarantees the
true optimum cannot change under such a rescaling.
Algorithm~\ref{alg:gaugefixed} is provably immune to this failure by
Proposition~\ref{prop:gaugefix}.

\subsection{What the First Singular Vector Represents}

\begin{theorem}
\label{thm:variance}
$v_1=\arg\max_{\lVert v\rVert=1}\lVert Fv\rVert_2^2$.
\end{theorem}

\begin{corollary}[Feature importance]
Entries of $v_1$ with large $|v_{1,x}|$ contribute strongly to the dominant
pattern; the sign indicates the direction of correlation.
\end{corollary}

\subsection{An Explicit NSW-Targeted Matching Algorithm}
\label{sec:algorithm2}

\begin{lemma}[Pointwise error bound]
\label{lem:err}
Assume $\lVert u_i\rVert=1$ for all $i$ (WLOG by Lemma~\ref{lem:gauge}). Let
$\Delta:=\sqrt{X-1}\,\sigma_2$. Then $|U_{ij}-\tilde U_{ij}|\le\Delta$ for all
$i,j$, where $\tilde U_{ij}:=(u_i\cdot v_1)(f_j\cdot v_1)$.
\end{lemma}
\begin{proof}
By Cauchy--Schwarz, $|U_{ij}-\tilde U_{ij}|\le\lVert f_j-(f_j\cdot v_1)v_1\rVert_2
=:\delta_j$. By Theorem~\ref{thm:eckartyoung} and orthonormality,
$\sum_j\delta_j^2=\sum_{\ell\ge2}\sigma_\ell^2\le(X-1)\sigma_2^2$; since each
$\delta_j^2$ is a nonnegative summand, $\delta_j\le\sqrt{X-1}\,\sigma_2$ for every
individual $j$.
\end{proof}

\begin{theorem}[Unconditional utilitarian welfare guarantee]
\label{thm:utilitarian}
Let $P^\ast$ be Algorithm~\ref{alg:gaugefixed}'s output and
$P_{\mathrm{sum}}\in\arg\max_P\sum_i\E[U_i\mid P]$ the true utilitarian optimum.
Then $\sum_i\E[U_i\mid P^\ast]\ge\sum_i\E[U_i\mid P_{\mathrm{sum}}]-2I\Delta$.
\end{theorem}
\begin{proof}
By the rearrangement inequality \citep{HardyLittlewoodPolya1952} (extended to
capacities by slot-expansion), Algorithm~\ref{alg:gaugefixed}'s sort maximizes
$\sum_i\tilde\E[\tilde U_i\mid P]$. By Lemma~\ref{lem:err},
$|\E[U_i\mid P]-\tilde\E[\tilde U_i\mid P]|\le\Delta$ for every $P$ and $i$; chain
the bound at $P^\ast$ and $P_{\mathrm{sum}}$.
\end{proof}

Algorithm~\ref{alg:gaugefixed} unconditionally guarantees near-optimal
\emph{utilitarian} welfare.

\begin{algorithm}[h]
\caption{Projected-NSW-optimal matching}\label{alg:nsw}
\begin{algorithmic}[1]
\Require $\hat a_1,\dots,\hat a_I$ (nonzero), $b_1,\dots,b_J$ (capacities $M_j$),
$\bar b=\frac1J\sum_jb_j$
\Ensure Allocation, or \textsc{Fail}
\State $J_+\gets\{j:b_j>\bar b\}$, $J_-\gets\{j:b_j<\bar b\}$;
$C_+\gets\sum_{j\in J_+}M_j$, $C_-\gets\sum_{j\in J_-}M_j$
\State $I_+\gets\{i:\hat a_i>0\}$, $I_-\gets\{i:\hat a_i<0\}$; $n_+\gets|I_+|$
\If{$n_+\ne C_+$} \Return \textsc{Fail} \EndIf
\State Sort $I_+$ desc.\ by $\hat a_i$; sort $J_+$'s slots desc.\ by $b_j$; pair
\State Sort $I_-$ desc.\ by $|\hat a_i|$; sort $J_-$'s slots asc.\ by $b_j$; pair
\end{algorithmic}
\end{algorithm}

\begin{theorem}[Exact optimality of Algorithm~\ref{alg:nsw}]
\label{thm:exact}
Let $\widetilde{\mathrm{NSW}}(P):=\prod_i\tilde g_i(P)$ when all $\tilde g_i(P)>0$,
else $0$. If a feasible deterministic $P$ with all $\tilde g_i(P)>0$ exists, then
$n_+=C_+$, Algorithm~\ref{alg:nsw} returns one such $P^\#$, and
$\widetilde{\mathrm{NSW}}(P^\#)=\max_P\widetilde{\mathrm{NSW}}(P)$; every such $P$
achieves this identical value. Otherwise Algorithm~\ref{alg:nsw} returns
\textsc{Fail} and $\widetilde{\mathrm{NSW}}(P)=0$ for every feasible deterministic
$P$.
\end{theorem}
\begin{proof}
For deterministic $\pi$, $\tilde g_i(\pi)=\hat a_i(b_{\pi(i)}-\bar b)>0$ forces
$b_{\pi(i)}>\bar b$ when $\hat a_i>0$ and $b_{\pi(i)}<\bar b$ when $\hat a_i<0$
(genericity rules out equality). So every all-positive-gain assignment uses only
$J_+$ for $I_+$ and only $J_-$ for $I_-$, requiring $n_+\le C_+$, $n_-\le C_-$;
since $n_++n_-=I=\sum_jM_j\ge C_++C_-$ with equality under genericity, both force
$n_+=C_+$. Given this, every all-positive-gain assignment uses all of $J_+$'s
capacity for $I_+$, so $\widetilde{\mathrm{NSW}}(P)$ is the product of a fixed
multiset. Otherwise no all-positive-gain assignment exists.
\end{proof}

\begin{proposition}[Equivalence]
\label{prop:equiv}
Whenever Algorithm~\ref{alg:nsw} does not return \textsc{Fail}, it returns the
same allocation as Algorithm~\ref{alg:gaugefixed}.
\end{proposition}

Algorithm~\ref{alg:gaugefixed} therefore already computes the exact projected-NSW
optimum at $O(J\log J+I\log I)$ cost.

\begin{lemma}[Multiplicative stability]
\label{lem:stab}
Let $y_i,\hat y_i>0$, $|y_i-\hat y_i|\le\Delta'$, $\mu:=\min_i\hat y_i\ge2\Delta'$.
Then $\prod_iy_i\ge\big(\prod_i\hat y_i\big)(1-2I\Delta'/\mu)$.
\end{lemma}

\begin{theorem}[Conditional NSW guarantee]
\label{thm:condnsw}
Let $P^\ast$ be Algorithm~\ref{alg:nsw}'s (equivalently
Algorithm~\ref{alg:gaugefixed}'s) output, $P_{\mathrm{opt}}\in\arg\max_P
\mathrm{NSW}(P\mid U)$. If $\mu^\ast:=\min_i\tilde g_i(P^\ast)\ge4\Delta$ and
$\gamma_0:=\min_i\mathrm{gain}_i(P_{\mathrm{opt}})\ge4\Delta$, then
\[
\mathrm{NSW}(P^\ast\mid U)\ge\mathrm{NSW}(P_{\mathrm{opt}}\mid U)
\Big(1-\frac{4I\Delta}{\gamma_0}-\frac{4I\Delta}{\mu^\ast}\Big).
\]
\end{theorem}
\begin{proof}
Apply Lemma~\ref{lem:stab} with $\Delta'=2\Delta$ at $P^\ast$, use
Theorem~\ref{thm:exact} to compare against $P_{\mathrm{opt}}$, apply
Lemma~\ref{lem:stab} again at $P_{\mathrm{opt}}$, and chain.
\end{proof}

$\mu^\ast$ is computable from the algorithm's own output \emph{before} any true outcome is observed; $\gamma_0$ is an unavoidable non-degeneracy requirement on the instance. We now prove the necessity of this requirement.

\begin{proposition}[Necessity]
\label{prop:necessity}
Suppose $\sigma_2=\cdots=\sigma_X=0$, $I=J$ with $M_j=1$, and $\hat a_i>0$ for all
$i$. If the $b_j$ are not all equal, $n_+=I>C_+<I$; therefore, Algorithm~\ref{alg:nsw}
returns \textsc{Fail} and $\mathrm{NSW}(P\mid U)=0$ for \emph{every} deterministic $P$, regardless of the mechanism.
\end{proposition}

When agents agree too strongly on which objects are best, i.e., when the regime approaches the low-effective-dimensionality (Section~\ref{sec:limitations}), no deterministic mechanism can guarantee positive NSW. Section~\ref{sec:illustration} shows this is not an edge case.

\subsection{Handling Non-Orthogonal Features}

\begin{proposition}
If features $x,x'$ are highly correlated across objects, $v_{1,x}\approx v_{1,x'}$ (up to scaling): SVD treats correlated features as a composite dimension automatically.
\end{proposition}

\subsection{Welfare Guarantees}

The statement $\mu^\ast\ge0$ is the individual rationality output of Algorithm~\ref{alg:gaugefixed}. The multiplicative NSW guarantee of Theorem~\ref{thm:condnsw} holds whenever $\mu^\ast\ge4\Delta$ and
$\gamma_0\ge4\Delta$. Furthermore, Proposition~\ref{prop:necessity} shows $\sigma_1\gg\sigma_2$ can simultaneously make $\gamma_0$ vanish, so the guarantee can fail entirely.

\section{Computational Complexity}
\label{sec:complexity}

\begin{theorem}
\label{thm:complexity}
Algorithm~\ref{alg:gaugefixed} (equivalently, whenever it does not fail,
Algorithm~\ref{alg:nsw}) runs in $O(\min(J^2X,JX^2)+I\log I+J\log J)$ time: SVD
computation \citep{GolubReinsch1970} dominates.
\end{theorem}

Direct NSW maximization is generally NP-hard \citep{ColeGkatzelis2018}; a
nonlinear solver requires $O((IJ)^3)$ operations per iteration with no
convergence guarantee, four to five orders of magnitude slower than
Algorithm~\ref{alg:gaugefixed} for $I=J=100$ in practice
(Section~\ref{sec:experiment}).

\section{Limitations and Societal Considerations}
\label{sec:limitations}

\subsection{Modeling Assumptions}

Theorem~\ref{thm:condnsw}'s guarantee requires $\sigma_1\gg\sigma_2$ \emph{and} $\gamma_0\ge4\Delta$; Proposition~\ref{prop:necessity} shows the first does not imply, and can actively undermine, the second. Practitioners should check both $\rho_1=\sigma_1^2/\sum_\ell\sigma_\ell^2$ (a market-level diagnostic) and the round-level quantities $n_+$ vs.\ $C_+$ and $\mu^\ast$ vs.\ $4\Delta$ (computable from Algorithm~\ref{alg:nsw}'s own execution) before trusting Theorem~\ref{thm:condnsw}'s guarantee for a specific instance.

\subsection{Strategic Considerations}

Proposition~\ref{prop:manip} proves that the mechanism is manipulable by a single agent. Several further strategic issues remain open: coordinated misreporting by multiple agents, platform-side manipulation of feature reports $f_j$, and adaptive learning of $v_1$ under repeated play.

\subsection{Computational and Ethical Considerations}

Algorithm~\ref{alg:gaugefixed} produces deterministic allocations; fairness criteria requiring randomization are not addressed here. On the equity front, the mechanism satisfies individual-level symmetry only. If demographic groups have systematically different preference distributions, $v_1$ can encode and propagate that asymmetry. Proposition~\ref{prop:necessity} shows the mechanism can fail hardest when a market's agents agree the most; a condition that can be tracked and flagged.

\section{Numerical Illustrations}
\label{sec:illustration}

\subsection{An Example}

Three agents, three products, two features, $M_j=1$:
\[
f_1=(7.65,1.82),\quad f_2=(5.45,3.62),\quad f_3=(3.42,1.93),
\]
\[
u_1=(8,3),\quad u_2=(6,7),\quad u_3=(5,4).
\]
The SVD of $F$ gives $\sigma_1\approx10.803$, $\sigma_2\approx1.833$
($\Delta\approx1.833$ for $X=2$) and $v_1$ with both entries of the same sign.
Gauge-fixing and projecting gives $\hat a_1,\hat a_2,\hat a_3$ all of the same
sign, while the objects' projected scores split as one above their mean and two
below.

By Theorem~\ref{thm:exact}, this is exactly the configuration of
Proposition~\ref{prop:necessity}: $n_+=3$ agents share one sign while only
$C_+=1$ object lies above the mean. Algorithm~\ref{alg:nsw} correctly returns \textsc{Fail}. Exhaustive search over all $3!=6$ possible assignments confirms that this is not a limitation of the diagnostic: the true NSW-optimal value, over every possible deterministic assignment, is exactly $0$.

Algorithm~\ref{alg:gaugefixed} (which does not check feasibility and always returns an assignment) matches agent~1 to object~1, agent~2 to object~3, and agent~3 to object~2, achieving utilitarian welfare $142.42$ against a true utilitarian optimum of $149.52$ ($95.2\%$), consistent with Theorem~\ref{thm:utilitarian}: the gap, $7.1$, is below the bound $2I\Delta\approx11.0$ and is strict $\mathrm{NSW}=0$ (agent~2 receives a gain $-16.21<0$), with $\mathrm{NSW}_{0.01}\approx0.666$ once that violation is priced at a floor of $\epsilon=0.01$ per the metric policy of
Section~\ref{sec:robustness}.

This example is degenerate by Proposition~\ref{prop:necessity}: no deterministic mechanism could have done better on these true preferences. A perturbation of $u_1\to(5,3)$ remains
infeasible by the same diagnostic ($n_+=0\ne C_+=1$). We do not describe the resulting reassignment as the mechanism "responding" to changed preferences, since it is better understood as an arbitrary tie-break among allocations.

\section{Experiment: SVD-Based Matching for Agentic AI Shopping}
\label{sec:experiment}

We instantiate the corrected mechanism (Algorithm~\ref{alg:gaugefixed}/
\ref{alg:nsw}, and a single $\mathrm{NSW}_\epsilon$ metric with mandatory disclosure, defined in Section~\ref{sec:robustness}) in a limited-inventory agentic-shopping scenario: ten AI shopping-agent personas, each reporting feature preferences over a product catalog on behalf of a distinct user, competing for ten scarce SKUs in a single drop. Every number below is computed, not just assumed.

\subsection{Motivation}

AI shopping agents that translate a user's natural-language needs into structured
feature preferences, then compete for limited inventory, are an active deployment
scenario: \citet{Allouah2025} evaluate how such agents select products, and
\citet{Zhu2025} show that when AI agents negotiate on their users' behalf without
a principled allocation mechanism, outcomes can be systematically imbalanced.

\subsection{Setup}

We generate $J=24$ synthetic wireless-headphone SKUs with $X=6$ features
(affordability, battery life, noise cancellation, sound quality, brand
reputation, comfort), each drawn from a latent quality-tier variable plus
idiosyncratic noise (affordability anti-correlates with tier; the other five
features positively correlate with it). On the full catalog,
$\sigma=(68.39,18.27,7.62,5.89,4.62,2.49)$, giving $\rho_1=0.912$ and effective rank $r_{\mathrm{eff}}=2.24$, comfortably above a $\rho_1\ge0.5$ deployment threshold, with $v_1$ loading positively on every feature.

Ten of the 24 SKUs (a diverse mix of premium, budget, and mid-tier products) are offered with one unit each, $I=J=10$. Restricted to these ten, $\sigma_1=44.86$, $\sigma_2=15.74$, giving $\rho_1=0.871$. This number is still above the threshold, but with $\Delta=\sqrt{X-1}\,\sigma_2=35.20$, a projection-error bound large relative to
individual agents' gains at this scale, a distinction the market-level $\rho_1$ alone does not reveal. Ten personas represent the feature-weight reports an LLM shopping assistant would plausibly produce from a short natural-language brief. Six are mainstream (all positive-leaning along $v_1$), four are contrarian (plausible shoppers who actively avoid "more-is-better" featurism), included specifically so that some agents disagree in sign along $v_1$.

\subsection{Scenario A: Realistic Personas, No Engineered Disagreement}

Using six mainstream personas and four other realistic-but-mild personas (all positive-leaning), gauge-fixed scores are all positive ($n_+=10$), while only $C_+=6$ of the ten SKUs lie above the mean projected score. By Theorem~\ref{thm:exact}, $n_+\ne C_+$ means that no deterministic allocation can satisfy individual rationality for every agent; Algorithm~\ref{alg:nsw} correctly detects this and returns \textsc{Fail}. This is Proposition~\ref{prop:necessity} occurring organically, from realistic personas.

\subsection{Scenario B: Disagreement, Tuned to Feasibility}

Replacing four mainstream personas with the contrarian four gives $n_+=C_+=6$ exactly, so Theorem~\ref{thm:exact} applies. Running Algorithm~\ref{alg:nsw}, one agent (budget\_student) has a negative true gain: $1/10$ IR violations, strict
$\mathrm{NSW}=0$, $\mathrm{NSW}_{0.01}=3.72\times10^{15}$. Utilitarian welfare is
$1467.57$ against a true optimum of $1480.76$ ($99.1\%$, consistent with
Theorem~\ref{thm:utilitarian}).

\textbf{The pre-deployment diagnostic correctly predicts this failure.} Before observing the outcome: $\mu^\ast=0.260$ against a required margin
$4\Delta=140.80$. In this case, $\mu^\ast$ is failing by three orders of magnitude, computable from the algorithm's own output before any true-utility outcome is known.

\textbf{What was actually achievable.} An exhaustive search over all
$10!=3{,}628{,}800$ assignments (16.5s) finds a true $\mathrm{NSW}_\epsilon$-optimum
with every gain positive (minimum gain $\gamma_0=12.27$, itself short of $4\Delta$), so Theorem~\ref{thm:condnsw} correctly declines to certify even this optimum in advance and  $\mathrm{NSW}_{0.01}=1.392\times10^{18}$. Algorithm~\ref{alg:nsw} achieves only $3.72\times10^{15}/1.392\times10^{18}\approx0.27\%$ of this true optimum: when the margin condition fails this badly, the shortfall is catastrophic, just as Theorem~\ref{thm:condnsw}'s conditional (not unconditional) nature warns.

\begin{remark}[Equivalence check]
The assignment produced by plain gauge-fixed sorted matching
(Algorithm~\ref{alg:gaugefixed}, implemented independently as a baseline) coincides exactly with Algorithm~\ref{alg:nsw}'s output on this instance. This provides a direct empirical confirmation of Proposition~\ref{prop:equiv}.
\end{remark}

\subsection{Manipulability, Tested Empirically}

We ask whether budget\_student can improve true utility by misreporting, holding other agents' reports fixed. This is an empirical test of Proposition~\ref{prop:manip} on a realistic instance. A random search over $200{,}000$ candidate unit-norm reports finds one achieving true gain $+24.06$, a gain of $+33.61$ over truthful reporting. Truthful reporting is not a best response for this agent here.

\subsection{Efficiency}

Algorithm~\ref{alg:nsw} runs in $7.6\,\mu s$ on average for this instance; exhaustive search takes $16.5$s, which is about $2.18\times10^6$ times slower, consistent with Theorem~\ref{thm:complexity}.

\subsection{Takeaways}

(1) The market-level $\rho_1$ diagnostic is necessary but not sufficient: a specific round can have materially worse $\rho_1$ and larger $\Delta$ than the broader catalog from which it is drawn. (2) The $n_+=C_+$ and $\mu^\ast\ge4\Delta$ checks, both computable from the algorithm's own output before any outcome is observed, correctly flagged both failure modes here before brute-force confirmation. (3) When a round fails these diagnostics, exact optimization is
often feasible for small, high-stakes drops. (4) A deployment should not assume truthful reporting by default: Proposition~\ref{prop:manip} is not an edge case.

\section{Robustness Across Random Market Instances}
\label{sec:robustness}

Section~\ref{sec:experiment} validates the mechanism on one carefully constructed instance. This section reports results pooled from 100 independently generated random market instances, comparing them against random allocation, serial dictatorship, and the true utilitarian optimum, establishing which findings above are typical rather than constructed. Every number below is computed from the accompanying simulation code; none is assumed.

\subsection{Setup}

Each of the 100 trials independently generates a 24-SKUs catalog using the generative model of Section~\ref{sec:experiment}, selects 10 SKUs for that trial's drop, and generates $I=10$ agent personas as a random mixture of mainstream and mission-driven types in the same qualitative proportions. Capacities are $M_j=1$,
so $I=J=10$ throughout.

\textbf{Mechanisms compared.} \emph{Random}: mean over 20 random permutations per trial. \emph{Serial Dictatorship} \citep{AbdulkadirogluSonmez1998}: agents arrive
in random order and each selects their most-preferred remaining SKU using
\emph{true} utilities. \emph{Our Method}: Algorithm~\ref{alg:gaugefixed}, whose
output coincides with Algorithm~\ref{alg:nsw}'s exact optimum whenever the latter
is feasible, checked computationally on every trial. \emph{Utilitarian-Optimal}:
computed exactly via the Hungarian algorithm.

\begin{definition}[The metric]
\label{def:metric}
For an allocation $P$, $\mathrm{gain}_i^\epsilon(P):=\max(\mathrm{gain}_i(P),
\epsilon)$ and $\mathrm{NSW}_\epsilon(P\mid U):=\prod_i\mathrm{gain}_i^\epsilon(P)$,
with $\epsilon=0.01$ fixed throughout. $\mathrm{NSW}_\epsilon$ is never reported
without its companion individual-rationality (IR) violation rate, since
$\mathrm{NSW}_\epsilon(P\mid U)\ge\mathrm{NSW}(P\mid U)$ always has equality only when no agent is clipped.
\end{definition}

\subsection{Results}

\textbf{Feasibility.} Algorithm~\ref{alg:nsw}'s exact diagnostic ($n_+=C_+$) passed in 22 of 100 trials ($22\%$) and is neither a near-universal nor a near-impossible condition.

\textbf{Per-agent gain distribution.} Figure~\ref{fig:utildist} pools every per-agent gain across all trials: $1{,}000$ observations under our method, $20{,}000$ under random allocation. The pooled mean gain is $19.44$ under our method and $-0.02$ under random allocation. This provides a direct check on the simulation itself, since the disagreement point is defined so that the true expected gain under random allocation is exactly zero. The fraction of IR-violating observations is $15.6\%$ under our method against $50.3\%$ under random allocation.

\begin{figure}[h]
\centering
\includegraphics[width=0.62\textwidth]{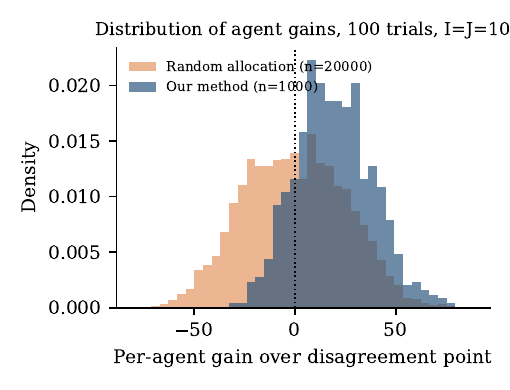}
\caption{Distribution of per-agent gains over the disagreement point, pooled
across 100 random trials.}
\label{fig:utildist}
\end{figure}

\textbf{Nash Social Welfare across mechanisms.} Figure~\ref{fig:nswcomp} and
Table~\ref{tab:nswtable} report mean $\log\mathrm{NSW}_\epsilon$ with 95\%
confidence intervals ($n=100$), each annotated with its mean IR-violation rate.

\begin{figure}[h]
\centering
\includegraphics[width=0.62\textwidth]{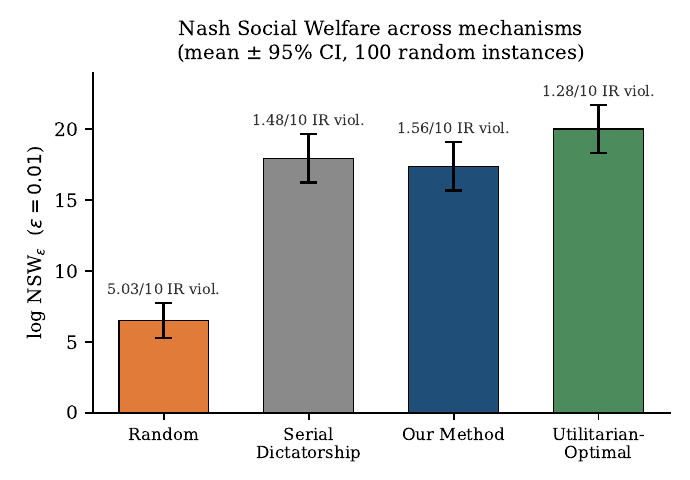}
\caption{Mean $\log\mathrm{NSW}_\epsilon$ across mechanisms, 100 random
instances, with 95\% CIs and mean IR-violation rate annotated on every bar.}
\label{fig:nswcomp}
\end{figure}

\begin{table}[h]
\centering
\begin{tabular}{@{}lrrr@{}}
\toprule
Mechanism & Mean $\log\mathrm{NSW}_\epsilon$ & 95\% CI half-width & Mean IR viol.\ (of 10) \\
\midrule
Random & 6.52 & $\pm1.23$ & 5.03 \\
Serial Dictatorship & 17.95 & $\pm1.72$ & 1.48 \\
Our Method & 17.37 & $\pm1.70$ & 1.56 \\
Utilitarian-Optimal & 20.01 & $\pm1.67$ & 1.28 \\
\bottomrule
\end{tabular}
\caption{Figures underlying Figure~\ref{fig:nswcomp}.}
\label{tab:nswtable}
\end{table}

\subsection{Discussion}

\textbf{Our method dominates random allocation} by a wide, non-overlapping
margin: roughly 11 log-points, with less than a third of the IR-violation rate.

\textbf{Our method does not beat serial dictatorship on this metric, on
average}: the two 95\% CIs overlap almost completely (17.37 vs.\ 17.95), and
serial dictatorship's point estimate is marginally higher. Serial dictatorship
uses each agent's true, unprojected utility directly, and carries no formal
guarantee at all; our method uses only a one-dimensional projection, runs in
$O(J\log J+I\log I)$ time rather than requiring $I$ sequential best-response
computations, and comes with Theorem~\ref{thm:condnsw}'s checkable conditional
guarantee. That the two perform comparably on average does not contradict
Proposition~\ref{prop:necessity}: Section~\ref{sec:experiment}'s Scenario~B is a
specific, diagnosable instance in which the gap is not small at all, and the
$1.56$-of-$10$ mean IR-violation rate here confirms the mechanism is not becoming
strategyproof or uniformly individually rational by virtue of averaging well.

\subsection*{Reproducibility}
All catalog and persona generation, mechanism implementations, and figure code
accompany this paper; regenerating both figures and Table~\ref{tab:nswtable} from
a fresh random seed is a single script invocation. Figures use Matplotlib
\citep{Hunter2007}; the Hungarian-algorithm baseline uses SciPy
\citep{Virtanen2020}.

\section{Discussion and Extensions}
\label{sec:discussion}

\subsection{When Does Low-Rank Structure Arise?}

\begin{proposition}[Sources of low-rank structure]
Low effective dimensionality arises from correlated features, a dominant
quality factor, constraint-induced (e.g.\ budget-driven) correlation, or market
equilibrium effects that equalize utility-per-dollar across objects.
\end{proposition}

In school-choice data, schools often exhibit strong positive correlations
between test scores across subjects, with the first principal component
explaining 70--80\% of variance; in housing, price per square foot is often
dominant, with the first component typically explaining 50--60\%.

\subsection{Extensions}

For $k>1$ principal components, project onto the top-$k$ singular vectors and
solve the resulting $k$-dimensional assignment problem via the Hungarian
algorithm or a sequential approach \citep{ChawlaHartlineMalecSivan2010}; this
trades some of Section~\ref{sec:methodology}'s one-dimensional guarantees for reduced approximation error. Group fairness constraints can be added directly to the NSW optimization at the cost of solving a constrained program rather than
sorting. An online variant recomputes $v_1$ periodically and assigns arriving
agents to the best available object under the current projection, suited to
rolling admissions. The framework extends to two-sided matching by computing an
SVD for each side and balancing both sides' projected scores.

\subsection{Scalability}

For $I,J>10^6$, randomized SVD algorithms compute approximate leading singular
vectors in $O(JX\log X)$ time rather than $O(J^2X)$; distributed computation and
incremental SVD updates are natural further extensions for markets that change
slowly over time.

\section{Conclusion}
\label{sec:conclusion}

We set out to determine whether a single spectral projection can turn a
multi-dimensional, feature-based matching problem into a fast, fair
one-dimensional sort. The answer is conditional. Algorithm~\ref{alg:nsw} is Nash-Social-Welfare-optimal within the space it actually searches (Theorem~\ref{thm:exact}), and Algorithm~\ref{alg:gaugefixed} computes this same
allocation at $O(N\log N)$ cost (Proposition~\ref{prop:equiv}), being unconditionally competitive on utilitarian welfare (Theorem~\ref{thm:utilitarian}). Its guaranty for true Nash Social Welfare is real but conditional (Theorem~\ref{thm:condnsw}), and we have shown that the condition is necessary: when a market's agents agree too strongly about which objects are best, no deterministic mechanism can guaranty a positive outcome (Proposition~\ref{prop:necessity}). We provided two diagnostics, computable before any outcome is observed, that detect this in advance both at the market level and at the level of a specific allocation round. We showed empirically (Section~\ref{sec:experiment}) that both diagnostics correctly
predicted a real failure before brute-force search confirmed it.

We also corrected the mechanism's incentive story. It is noise-stable
(Theorem~\ref{thm:noisestab}) but not strategyproof
(Proposition~\ref{prop:manip}), and we exhibited an explicit profitable
misreport, independently rediscovered by a blind search on realistic data. This matters more, not less, as this class of mechanism moves toward deployment in agentic commerce, where the agents doing the reporting are themselves optimizing systems.

Our 100-instance robustness study (Section~\ref{sec:robustness}) shows that the mechanism robustly dominates random allocation but does not, on average, outperform serial dictatorship in terms of Nash Social Welfare. What the mechanism offers instead of a larger average margin is speed, determinism, and a formal, checkable guarantee that all outperform serial dictatorship.

\subsection{Limitations and Future Work}

The additive-utility assumption rules out complementarities. It is natural to extend the model to interaction terms,  but such extensions will increase effective dimensionality. Instead, we can close the gap between the projected and true NSW optima when the margin diagnostic fails via the rank-$k$ extension of Section~\ref{sec:discussion}, or a hybrid that falls back to optimization only for the specific agents a diagnostic flags. A mechanism with a real incentive guaranty, even
a weak one, remains an open problem: Theorem~\ref{thm:zhou} says exact efficiency, truthfulness, and symmetry cannot all be achieved at the same time, but it says nothing about how close a corrected mechanism could come. Finally, validating the diagnostics themselves on real deployment data across school choice, labor markets, course allocation, and agentic commerce is necessary before any of these guarantees should be trusted outside simulation.

\bibliographystyle{ACM-Reference-Format}
\bibliography{matching}

\end{document}